\documentclass[aip,rsi,preprint,groupedaddress,showmsc,superscriptaddress,floatfix]{revtex4}
\usepackage{amsmath,amsfonts,amssymb,caption,hyperref,color,epsfig,graphics,graphicx,latexsym,mathrsfs,revsymb,theorem,url,verbatim,epstopdf}

\usepackage[center]{titlesec}
\titleformat{\section}[hang]{\Large\bfseries\filcenter}{}{1em}{}
\titleformat{\subsection}[hang]{\bfseries}{}{1em}{}
\usepackage{amsmath,amsfonts,amssymb,caption,hyperref,color,epsfig,graphics,graphicx,latexsym,mathrsfs,revsymb,theorem,url,verbatim,epstopdf,cleveref}
\hypersetup{colorlinks,linkcolor={blue},citecolor={blue},urlcolor={red}}

\newtheorem{definition}{Definition}
\newtheorem{proposition}[definition]{Proposition}
\newtheorem{lemma}[definition]{Lemma}

\newtheorem{theorem}[definition]{Theorem}

\newtheorem{question}[definition]{Question}

\def\squareforqed{\hbox{\rlap{$\sqcap$}$\sqcup$}}
\def\qed{\ifmmode\squareforqed\else{\unskip\nobreak\hfil
\penalty50\hskip1em\null\nobreak\hfil\squareforqed
\parfillskip=0pt\finalhyphendemerits=0\endgraf}\fi}
\def\endenv{\ifmmode\;\else{\unskip\nobreak\hfil
\penalty50\hskip1em\null\nobreak\hfil\;
\parfillskip=0pt\finalhyphendemerits=0\endgraf}\fi}
\newenvironment{proof}{\noindent \textbf{{Proof.~} }}{\qed}
\def\Dbar{\leavevmode\lower.6ex\hbox to 0pt
{\hskip-.23ex\accent"16\hss}D}
\def\bpf{\begin{proof}}
\def\epf{\end{proof}}

\newcommand{\ket}[1]{|{#1}\rangle}

\newcommand{\abs}[1]{\left\lvert {#1} \right\rvert}

\newcommand{\red}{\textcolor{red}}

\newcommand{\tbc}{\red{TO BE CONTINUED}.~}

\newcommand{\nc}{\newcommand}

\def\bea{\begin{eqnarray}}
\def\eea{\end{eqnarray}}
\def\beq{\begin{equation}}
\def\eeq{\end{equation}}
\def\bal{\begin{aligned}}
\def\eal{\end{aligned}}
\def\bma{\begin{bmatrix}}
\def\ema{\end{bmatrix}}

\def\diag{\mathop{\rm diag}}

\def\dg{\dagger}

\def\ra{\rightarrow}

\def\a{\alpha}
\def\b{\beta}

\def\i{\iota}

\def\m{\mu}

\def\p{\pi}

\def\o{\omega}

\nc{\bbA}{\mathbb{A}} \nc{\bbB}{\mathbb{B}} \nc{\bbC}{\mathbb{C}}
\nc{\bbD}{\mathbb{D}} \nc{\bbE}{\mathbb{E}} \nc{\bbF}{\mathbb{F}}
\nc{\bbG}{\mathbb{G}} \nc{\bbH}{\mathbb{H}} \nc{\bbI}{\mathbb{I}}
\nc{\bbJ}{\mathbb{J}} \nc{\bbK}{\mathbb{K}} \nc{\bbL}{\mathbb{L}}
\nc{\bbM}{\mathbb{M}} \nc{\bbN}{\mathbb{N}} \nc{\bbO}{\mathbb{O}}
\nc{\bbP}{\mathbb{P}} \nc{\bbQ}{\mathbb{Q}} \nc{\bbR}{\mathbb{R}}
\nc{\bbS}{\mathbb{S}} \nc{\bbT}{\mathbb{T}} \nc{\bbU}{\mathbb{U}}
\nc{\bbV}{\mathbb{V}} \nc{\bbW}{\mathbb{W}} \nc{\bbX}{\mathbb{X}}
\nc{\bbZ}{\mathbb{Z}}

\nc{\bA}{{\bf A}} \nc{\bB}{{\bf B}} \nc{\bC}{{\bf C}}
\nc{\bD}{{\bf D}} \nc{\bE}{{\bf E}} \nc{\bF}{{\bf F}}
\nc{\bG}{{\bf G}} \nc{\bH}{{\bf H}} \nc{\bI}{{\bf I}}
\nc{\bJ}{{\bf J}} \nc{\bK}{{\bf K}} \nc{\bL}{{\bf L}}
\nc{\bM}{{\bf M}} \nc{\bN}{{\bf N}} \nc{\bO}{{\bf O}}
\nc{\bP}{{\bf P}} \nc{\bQ}{{\bf Q}} \nc{\bR}{{\bf R}}
\nc{\bS}{{\bf S}} \nc{\bT}{{\bf T}} \nc{\bU}{{\bf U}}
\nc{\bV}{{\bf V}} \nc{\bW}{{\bf W}} \nc{\bX}{{\bf X}}
\nc{\bZ}{{\bf Z}}

\nc{\bmA}{{\bm A}} \nc{\bmB}{{\bm B}} \nc{\bmC}{{\bm C}}
\nc{\bmD}{{\bm D}} \nc{\bmE}{{\bm E}} \nc{\bmF}{{\bm F}}
\nc{\bmG}{{\bm G}} \nc{\bmH}{{\bm H}} \nc{\bmI}{{\bm I}}
\nc{\bmJ}{{\bm J}} \nc{\bmK}{{\bm K}} \nc{\bmL}{{\bm L}}
\nc{\bmM}{{\bm M}} \nc{\bmN}{{\bm N}} \nc{\bmO}{{\bm O}}
\nc{\bmP}{{\bm P}} \nc{\bmQ}{{\bm Q}} \nc{\bmR}{{\bm R}}
\nc{\bmS}{{\bm S}} \nc{\bmT}{{\bm T}} \nc{\bmU}{{\bm U}}
\nc{\bmV}{{\bm V}} \nc{\bmW}{{\bm W}} \nc{\bmX}{{\bm X}}
\nc{\bmZ}{{\bm Z}}

\nc{\cA}{{\cal A}} \nc{\cB}{{\cal B}} \nc{\cC}{{\cal C}}
\nc{\cD}{{\cal D}} \nc{\cE}{{\cal E}} \nc{\cF}{{\cal F}}
\nc{\cG}{{\cal G}} \nc{\cH}{{\cal H}} \nc{\cI}{{\cal I}}
\nc{\cJ}{{\cal J}} \nc{\cK}{{\cal K}} \nc{\cL}{{\cal L}}
\nc{\cM}{{\cal M}} \nc{\cN}{{\cal N}} \nc{\cO}{{\cal O}}
\nc{\cP}{{\cal P}} \nc{\cQ}{{\cal Q}} \nc{\cR}{{\cal R}}
\nc{\cS}{{\cal S}} \nc{\cT}{{\cal T}} \nc{\cU}{{\cal U}}
\nc{\cV}{{\cal V}} \nc{\cW}{{\cal W}} \nc{\cX}{{\cal X}}
\nc{\cZ}{{\cal Z}}

\begin{document}


\title{Real Entries of Complex Hadamard Matrices and Mutually Unbiased Bases in Dimension Six}

\author{Mengfan Liang}
\affiliation{School of Mathematics and Systems Science, Beihang University, Beijing 100191, China}

\author{Mengyao Hu}\email[]{humengyao@bjfu.edu.cn (corresponding author)}
\affiliation{School of Mathematics and Systems Science, Beihang University, Beijing 100191, China}

\author{Yize Sun}\email[]{sunyize@buaa.edu.cn (corresponding author)}
\affiliation{School of Mathematics and Systems Science, Beihang University, Beijing 100191, China}

\author{Lin Chen}\email[]{linchen@buaa.edu.cn (corresponding author)}
\affiliation{School of Mathematics and Systems Science, Beihang University, Beijing 100191, China}
\affiliation{International Research Institute for Multidisciplinary Science, Beihang University, Beijing 100191, China}

\begin{abstract}
We investigate the number of real entries of an $n\times n$ complex Hadamard matrix (CHM). We analytically derive the numbers when $n=2,3,4,6$. In particular, the number can be any one of $0-22,24,25,26,30$ for $n=6$. We apply our result to the existence of four mutually unbiased bases (MUBs) in dimension six, which is a long-standing open problem in quantum physics and information. We show that if four MUBs containing the identity matrix exists then the real entries in any one of the remaining three matrices does not exceed $22$.
\end{abstract}

\date{\today}

\maketitle


\tableofcontents

\section{Introduction}
\label{sec:intro}

In this paper we shall refer to the complex Hadamard matrix (CHM) as a square matrix having elements of modulus one, and pairwise orthogonal row and column vectors. We propose and investigate the following question. 
\begin{question}
\label{qu:main}
How many real entries are there in a given CHM?	
\end{question}
We partially answer this question, and leave the complete answer as an open problem. We construct preliminary results in Lemma \ref{le:linearalg}, \ref{le:properties} and \ref{le:mubtrio}. As the first main result of this paper, we shall characterize the number of real entries of $n\times n$ CHMs with $n=2$ and $n=3,4,6$ in Lemma \ref{le:s2} and Theorem \ref{thm:s3456}, respectively. We show that the number for $n=3,4,6$ belongs to the set $\{0,1,2,3,4,5,6\}$, $\{0-10,12,16\}$, and $\{0-22,24,25,26,30\}$, respectively. The proof for $n=6$ is the most complicated, and we explain it 
in Proposition \ref{pp:0-2224252630inCS6} and \ref{pp:3133-36notinCS6}. The latter is  based on Lemma \ref{le:23notincs6-a1=0}, \ref{le:23notincs6-a1=1}, and \ref{le:23notincs6-a1=2}. We shall analytically construct $6\times6$ CHMs containing exactly $n$ real entries with the above-mentioned integer $n$, and excluded the CHMs containing exactly $n$ real entries with $n$ not mentioned above. We apply our results to a conjecture on the existence of so-called four six dimensional mutually unbiased basis (MUBs) from quantum physics and information. In Theorem \ref{thm:app}, we show that the CHM as a member of an MUB trio has at most $22$ real elements.

Characterizing the $n\times n$ CHM, especially the case $n=6$ is a basic problem in algebra and quantum information theory \cite{Sz12,mb15,Goyeneche13, Turek2016A,NICOARA2019143, Sz2008Parametrizing}. It is known that the real Hadamard matrix has order $2$ or $4k$, and whether it exists for any integer $k$ has been an open problem for more than one century \cite{Djokovic2014Some}. Characterizing the $n\times n$ CHM is a more complex problem, though it is done for $n=2,3,4,5$ \cite{mub09}. Finding the real entries is thus an operational method of studying CHM. There is no systematic result as far as we know. On the other hand, the quantum state is a unit vector in linear algebra. Two states in the $d$-dimensional Hilbert space $\bbC^d$ are MU when their inner product has modulus ${1\over\sqrt d}$. We say that two orthonormal basis are MU when their elements are all MU. If $n$ orthonormal bases $\cB_1,\cB_2,..,\cB_n$ in $\bbC^d$ are pairwise MU then we say that they are MUBs in $\bbC^d$. By regarding $\cB_j$ as the column vectors of a matrix $B_j$, then the latter is a unitary matrix. For simplicity we shall also refer to $B_1,B_2,..,B_n$ as MUBs. Evidently $UB_1,UB_2,..,UB_n$ are still MUBs for any unitary matrix $U$. If we choose $U=B_1^\dg$ then $UB_1$ is the $d\times d$ identity matrix denoted as $I_d$. Since it is MU to $UB_2,..,UB_n$, any one of these matrices is a CHM multiplied by ${1\over\sqrt d}$. For $d=6$, it has been a long-standing open problem whether four MUBs $I_6,V,W,X$ exist. If it exists then we refer to $V,W,X$ as an MUB trio. In spite of much efforts devoted to the open problem in the past decades \cite{Boykin05,bw08,bw09,jmm09,bw10,deb10,wpz11,mw12ijqi,mw12jpa135307,rle11,mw12jpa102001,mpw16, Chen2017Product,Chen2018Mutually, Designolle2018Quantifying}, there has been little understand of the MUB trio, though it is believed not to exist. Theorem \ref{thm:app} gives an upper bound on the number of real entries of $6\times6$ CHMs of an MUB trio, if it exists. As far as we know, this is the first time the real entry of CHM has been investigated for studying the existence problem of four six-dimensional MUBs. It provides theoretically novel view to the MUB existence problem and related problems in quantum information, such as the understanding of general unitary matrices, tensor rank and unextendible product basis \cite{Baerdemacker2017The,Chen2018The,kai=lma}.

The rest of this paper is structured as follows. In Sec. \ref{sec:pre} we construct the notion of CHMs, equivalence and complex equivalence of matrices, characterizing of CHMs with few imaginary entries and preliminary results from linear algebra. In Sec. \ref{sec:mainresult} we introduce the notations on the number of real entries of CHMs, and the main result of this paper, namely characterizing the number of real entries of $3\times3$, $4\times4$, and $6\times6$ CHMs respectively. We provide the proof details in Sec. \ref{sec:6x6proof}. We conclude in Sec. \ref{sec:con}.

\section{Preliminaries}
\label{sec:pre}

In this section we introduce the fundamental notations and facts we use in this paper. We start by reviewing the complex Hadamard matrices.

\begin{definition}
\label{df:hada}
We refer to the $n\times n$ complex Hadamard matrix (CHM) $H_n=[u_{ij}]_{i,j=1,...,n}$ as a  matrix with orthogonal row vectors and entries of modulus one. That is,
\begin{eqnarray}
H_n^\dag H_n = nI_n,
\quad\quad
\abs{u_{ij}} = 1.
\end{eqnarray}
\qed	
\end{definition}

To find out the connection between different CHMs, we define the equivalence and complex equivalence. 

\begin{definition}
\label{df:skewper}
(i) We refer to the monomial unitary matrix as a unitary matrix each of whose row and columns has exactly one nonzero entry. The entry has modulus one. Let $\cM_n$ be the set of $n\times n$ monomial unitary matrices.

(ii) We say that two $n\times n$ matrices $U$ and $V$ are complex equivalent when $U=PVQ$ where $P,Q\in\cM_n$. If $P,Q$ are both permutation matrices then we say that $U,V$ are equivalent
\qed	
\end{definition}

Evidently if $U,V$ are equivalent then they are complex equivalent, and the converse fails. The number of real entries of a CHM may be changed under complex equivalence, while it is unchanged under equivalence. For example, it is straightforward to show that any $n\times n$ CHM is complex equivalent to a CHM containing at least $2n+1$ entry one. They are in the first column and row of the CHM. Nevertheless, investigating the real entries of a general CHM is a complex problem. 

In the following lemma we introduce useful results from linear algebra.
\begin{lemma}
\label{le:linearalg}	
(i) Suppose $a+b+c=0$ with complex numbers $a,b,c$ of modulus one. Then $(a,b,c)\propto(1,\o,\o^2)$ or $(1,\o^2,\o)$ with $\o=e^{2\p i\over 3}$.

(ii) Suppose $a+b+c+d=0$ with complex numbers $a,b,c,d$ of modulus one. Then $a=-b,-c$ or $-d$. 
\end{lemma}
\begin{proof}
Assertion (i) and (ii) can be proven straightforwardly. 
\end{proof}

To investigate Question \ref{qu:main}, we construct a few properties of CHMs in Lemma \ref{le:properties}. It is one of the main tools we use for proving our main result in Theorem \ref{thm:s3456}.

\begin{lemma}
\label{le:properties}
Suppose $H_n$ is an $n\times n$ CHM.

(i) If $P,Q\in\cM_n$, then $PH_nQ$ and $H_n^T$ are both $n\times n$ CHMs.

(ii.a) If the first row of $H_n$ is real, then the second row of $H_n$ does not have exactly one imaginary entry. 

(ii.b) Furthermore if $n$ is even then the second row of $H_n$ does not have exactly three imaginary entries. Equivalently, $H_n$ has no submatrix $\bma 
r_1&r_2&r_3&r_4&r_5&i_1\\
r_6&r_7&r_8&i_2&i_3&r_9\\
\ema$ with real $r_j$ and imaginary $i_k$.

(ii.c) Furthermore if $n$ is even and the second row of $H_n$ has exactly two imaginary entries, then they are equal or opposite numbers. Up to equivalence the first two rows of $H_n$ are 
\begin{eqnarray}
\bma
1&1& v_{{n\over2}-1} & v_{{n\over2}-1} \\
x&-x& v_{{n\over2}-1} & -v_{{n\over2}-1}\\
\ema,	
\end{eqnarray}
where $x$ is an imaginary number of modulus one, and $v_m$ is the $m$-dimensional vector of element one.

(ii.d) Furthermore if $n=6$ and the second and third rows of $H_6$ have both exactly two imaginary entries, then up to equivalence the first three rows of $H_6$  are one of the following four matrices.
\begin{eqnarray}
\label{eq:row123-1}
&&
H_{61}=
\bma
1&1& 1&1&1&1 \\
i&-i& 1&1&-1&-1\\
i& 1 & -i& -1& -1& 1\\
\ema,
\\&&\label{eq:row123-2}
H_{62}=
\bma
1&1& 1&1&1&1 \\
-i&i& 1&1&-1&-1\\
-i& 1 & i& -1& -1& 1\\
\ema,
\\&&\label{eq:row123-3}
H_{63}=
\bma
1&1& 1&1&1&1 \\
\o^2 &-\o^2& 1&1&-1&-1\\
1& -1& \o^2 & 1 & -\o^2 & -1\\
\ema,
\\&&\label{eq:row123-4} 
H_{64}=\bma
1&1& 1&1&1&1 \\
\o &-\o & 1&1&-1&-1\\
1& -1& \o & 1 & -\o & -1\\
\ema,
\end{eqnarray}
where $\o:=e^{{2\p i\over3}}$, $H_{61}=H_{62}^*$ and $H_{63}=H_{64}^*$.

(iii) If $n$ is odd then $H_n$ has no two real columns or two real rows. 

(iv) If $n\equiv2(\mod4)$ then $H_n$ has no three real columns or three real rows.

(v) $H_6$ does not have a $4\times3$ or $3\times4$ real submatrix. 

(vi) If $H_6$ has $n(\ge3)$ rows each of which has exactly one imaginary entry, then the entries are in different columns of $H_6$. Further, the entries are $i$ or $-i$. 

(vii) $H_6$ has neither of the following two three rows.
\begin{eqnarray}
&&
\bma	
r_1 & r_2 & r_3 & r_4 & i_1 & c_1 \\
r_5 & r_6 & r_7 & r_8 & i_2 & c_2 \\
r_9 & r_{10} & r_{11} & c_3 & r_{12} & c_4 \\
\ema,
\quad\quad
\bma	
r_1 & r_2 & r_3 & r_4 & r_5 & c_1 \\
r_6 & r_7 & r_8 & r_9 & r_{10} & c_2 \\
r_{11} & r_{12} & r_{13} & c_3 & c_4 & c_3 \\
\ema,
\end{eqnarray} 
where the entry $r_j$ is real, $i_k$ is imaginary and $c_l$ is complex.

(viii) If $H_6$ has the submatrix 
$\bma 
i_1&i_2&r_1&r_2&r_3&r_4\\
i_3&r_5&i_4&r_6&r_7&r_8\\
\ema$ with real $i_j$ and imaginary $r_k$, then $i_1,i_3$ are equal or opposite. Further if $w_2,w_3$ are equal or opposite then they are $i$ or $-i$.
\end{lemma}
\begin{proof}
Assertion (i),(ii.a)-(ii.c) and (iii) follow from the fact that any two row vectors of a CHM are orthogonal. 

(ii.d) It follows 
from (ii.c) that up to equivalence, the first two rows of $H_6$ are $\bma
1&1& 1&1&1&1 \\
x&-x& 1&1&-1&-1\\
\ema,$ where $x$ is an imaginary number of modulus one. Since the third row of $H_6$ has exactly two imaginary entries, they are $y,-y$ using the orthogonality of first and third row of $H_6$. Up to equivalence the first three rows of $H_6$ have four cases.
\begin{eqnarray}
\label{eq:M1=matrix}
&&
M_1=
\bma
1&1& 1&1&1&1 \\
x&-x& 1&1&-1&-1\\
y& -y & a_1& b_1& c_1& d_1\\
\ema,	
\\&&
\label{eq:M2=matrix}
M_2=
\bma
1&1& 1&1&1&1 \\
x&-x& 1&1&-1&-1\\
y& a_2 & -y& b_2& c_2& d_2\\
\ema,
\\&&
\label{eq:M3=matrix}
M_3=
\bma
1&1& 1&1&1&1 \\
x&-x& 1&1&-1&-1\\
a_3& b_3& y& c_3 & -y& d_3\\
\ema,
\\&&
\label{eq:M4=matrix}
M_4=
\bma
1&1& 1&1&1&1 \\
x&-x& 1&1&-1&-1\\
a_4& b_4& y& -y & c_4& d_4\\
\ema,
\end{eqnarray}
where $y$ is an imaginary number of modulus one. Further 
\begin{eqnarray}
\label{eq:ajbjcjdj}	
\{a_j,b_j,c_j,d_j\}=\{1,1,-1,-1\}, 
\end{eqnarray}
because row $1$ and $3$ of $M_j$ are orthogonal for $j=1,2,3,4$.  Note that row $2$ and $3$ of $M_j$ are also orthogonal. Using \eqref{eq:M1=matrix}-\eqref{eq:M4=matrix} we have
\begin{eqnarray}
\label{eq:M1}
&& 2x^*y+a_1+b_1-c_1-d_1=0,
\\&& \label{eq:M2}
x^*y-x^*a_2-y+b_2-c_2-d_2=0,
\\&&
\label{eq:M3}
x^*a_3-x^*b_3+2y+c_3-d_3=0,
\\&&
\label{eq:M4}
x^*a_4-x^*b_4-c_4-d_4=0.
\end{eqnarray}
Recall that $x,y$ are imaginary numbers of modulus one, and $\{a_j,b_j,c_j,d_j\}=\{1,1,-1,-1\}$ for $j=1,2,3,4$ in \eqref{eq:ajbjcjdj}. We shall use these facts in the following arguments of solving \eqref{eq:M1}-\eqref{eq:M4}.

To solve \eqref{eq:M1}, the only possibility is $xy=\pm1$ and $a_1+b_1-c_1-d_1=-4,0$ or $4$. So \eqref{eq:M1} has no solution, and $M_1$ in \eqref{eq:M1=matrix} does not exist. 

To solve \eqref{eq:M2}, the only possibility is $b_2-c_2-d_2=\pm1$. Lemma \ref{le:linearalg} (ii) implies that $x^*y+b_2-c_2-d_2=-x^*a_2-y=0$. So $y=\pm x$ and $a_2=\pm1$. Multiplying the unitary $\diag(1,1,-1)$ on the lhs of $M_2$ in \eqref{eq:M2=matrix}, we may assume $y=x$ and $a_2=1$. So $x=y=\pm i$, $b_2=c_2=-1$ and $d_2=1$. We have obtained that $M_2$ is one of \eqref{eq:row123-1} and \eqref{eq:row123-2}.

To solve \eqref{eq:M3}, the only possibility is $a_3=-b_3$. So $c_3=-d_3$. Eq. \eqref{eq:M3} becomes $x^*a_3+y+c_3=0$.
Multiplying the unitary $\diag(1,1,-1)$ on the lhs of $M_3$ in \eqref{eq:M3=matrix}, we may assume $c_3=1$. Lemma \ref{le:linearalg} (i) implies that $(x^*a_3,y)=(\o,\o^2)$ or $(\o^2,\o)$. We have four solutions for $M_3$ as follows.
\begin{eqnarray}
&& M_{31}=
\bma
1&1& 1&1&1&1 \\
\o^2 &-\o^2& 1&1&-1&-1\\
1& -1& \o^2 & 1 & -\o^2 & -1\\
\ema,
\\&& M_{32}=
\bma
1&1& 1&1&1&1 \\
\o &-\o & 1&1&-1&-1\\
1& -1& \o & 1 & -\o & -1\\
\ema,
\\&& M_{33}=
\bma
1&1& 1&1&1&1 \\
-\o^2&\o^2& 1&1&-1&-1\\
-1& 1& \o^2& 1 & -\o^2& -1\\
\ema,
\\&& M_{34}=
\bma
1&1& 1&1&1&1 \\
-\o &\o & 1&1&-1&-1\\
-1& 1& \o & 1 & -\o & -1\\
\ema.	
\end{eqnarray}
One can show that $M_{31},M_{33}$ are equivalent, and $M_{32},M_{34}$ are also equivalent. We have obtained that $M_3$ is one of \eqref{eq:row123-3} and \eqref{eq:row123-4}.

Finally \eqref{eq:M4} has no solution due to Lemma \ref{le:linearalg} (ii).

(iv) We prove the assertion by contradiction. Let $n=4k+2$. If $H_n$ has three real columns, then up to equivalence we can assume that the first three columns of $H_n$ are $\bma v & v & u \\ v & -v & w \ema$ where $v$ is a $(2k+1)$-dimensional vector of element one, and $u,w$ are both $(2k+1)$-dimensional vectors of element one or minus one. Since the three column vectors are pairwise orthogonal, we obtain that $v$ is orthogonal to $u$ and $w$. It is a contradiction with the fact that $2k+1$ is odd. So $H_n$ having three real columns does not exist. Using the matrix transposition, we can show that $H_n$ having three real rows also does not exist.

(v) We prove the assertion by contradiction. Suppose $H_6$ has a $4\times3$ real submatrix. Up to complex equivalence we may assume that the first three rows of $H_6$ have real entries, except that the $2\times2$ submatrix $M$ in the left upper corner of $H_6$ may have imaginary elements. Lemma \ref{le:properties} (ii.c) implies that $M=\bma x&-x\\y &-y\ema$. So the first two row vectors of $H$ are not orthogonal. It is a contradiction with the fact that $H_6$ is a CHM. 

(vi) The first part of the assertion follows from (v). We prove the second part. Let $a_j$ be the entry in the $j$'th row for $j=1,..,n$ and $n\ge3$. Assertion (ii.c) implies that $a_1^*=\pm a_2=\pm a_3$, and $a_2^*=\pm a_3$. So $a_j=\pm i$ for $j=1,2,3$.

(vii) We prove the assertion for the first matrix by contradiction, and one can similarly prove the assertion for the second matrix. 

Suppose the first three rows of $H_6$ are 
\begin{eqnarray}
\label{eq:r1r2r3}
\bma	
r_1 & r_2 & r_3 & r_4 & i_1 & c_1 \\
r_5 & r_6 & r_7 & r_8 & i_2 & c_2 \\
r_9 & r_{10} & r_{11} & i_3 & r_{12} & c_3 \\
\ema,
\end{eqnarray}
where $r_j$ is real, $i_j$ is imaginary and $c_j$ is complex. It follows from assertion (v) that $c_1^*c_3$ or $c_2^*c_3$ is imaginary. By permuting row $1$ and $2$ of $H_6$ we may assume that $c_2^*c_3$ is imaginary. So row $2$ and $3$ of $H_6$ are not orthogonal. We have a contradiction, and have proven that $H_6$ does not have the first three rows in \eqref{eq:r1r2r3}.

(viii) The assertion follows from assertion (ii.b).
\end{proof}

Finally we review the following fact from \cite[Lemma 11]{Chen2017Product}. It gives the necessary condition by which a $6\times6$ CHM is a member of some MUB trio.
\begin{lemma}
\label{le:mubtrio}
Any MUB trio does not have the $6\times6$ CHM containing a $3\times2$ real submatrix.
\end{lemma}

\section{Main result}
\label{sec:mainresult}

In this section we investigate Question \ref{qu:main}, and introduce the main result of this paper. For this purpose we construct the following definition. 

\begin{definition}
Let $\cR(H_n)$ be the number of real entries of a given $n\times n$ CHM $H_n$, and $\cS_n$ the set of all possible numbers for a given $n$. That is $H_n$ has exactly $\cR(H_n)$ real entries and $n^2-\cR(H_n)$ non-real entries. So
\begin{eqnarray}
\cR(H_n)\in[0,n^2],
\quad\quad
\cS_n\subseteq\{0,1,...,n^2\}.
\end{eqnarray}	
\qed
\end{definition}

To demonstrate the definition, we present the observation on $2\times2$ CHMs.
\begin{lemma}
\label{le:s2}	
$\cS_2=\{0,1,2,4\}.	
$
\end{lemma}
\begin{proof}
We investigate $\cS_2$ by constructing the following $2\times2$ CHMs.
\begin{eqnarray}
H_{21}&=&\bma 1 & 1\\1&-1\ema,
\quad
H_{22}=\bma 1 & i\\1&-i\ema,
\\
H_{23}&=&\bma e^{{\p i\over4}} & e^{{\p i\over4}}i\\1&-i\ema,
\quad	
H_{24}=\bma i& i\\i&-i\ema.
\end{eqnarray}
We have
$
\cR(H_{21})=4,
\cR(H_{22})=2,
\cR(H_{23})=1,
\cR(H_{24})=0.
$
Suppose there is a CHM $
H_2=
\bma
u_{11} & u_{12} \\
u_{21} & u_{22}
\ema$	
satisfying $\cR(H_2)=3$. So three of $u_{11}, u_{12},
u_{21},u_{22}$ are real. They satisfy
$
u_{11}^*u_{12}+u_{21}^*u_{22}=0,
\abs{u_{ij}}=1.
$
The equation has no solution, so $\cR(H_2)\ne3$. So $2\times2$ CHMs may have $0,1,2,$ or $4$ real entries. We have proven
the assertion. 	
\end{proof}

Let $x_{j,k}=e^{{2\p i \over d}jk}$ and $\ket{x_j}=[x_{j,0},x_{j,1},...,x_{j,d-1}]^T\in\bbC^d$. One can verify that the set of $d$ vectors $\ket{x_0},\ket{x_1},...,\ket{x_{d-1}}$ is an orthonormal basis in $\bbC^d$. Hence the matrix $[\ket{x_0},\ket{x_1},...,\ket{x_{d-1}}]$ is a $d\times d$ CHM. Now we present the first main result of this paper. It characterizes $\cS_n$ for general $n$, and explicitly counts the number of real entries of $n\times n$ CHMs when $n=3,4,6$.
\begin{theorem}
\label{thm:s3456}	
(i) For any positive integer $n$ we have $\{0,1,...,n\}\subseteq\cS_n$. 

(ii) For any odd number $n$ we have $\{n+1,n+2,...,2n-1\}\subseteq\cS_n$.

(iii) $\cS_3=\{0,1,2,3,4,5,6\}.	
$

(iv) $\cS_4=\{0-10,12,16\}.	
$

(v) $\cS_6=\{0-22,24,25,26,30\}.	$
\end{theorem}
\begin{proof}
(i)
Consider the $n\times n$ CHM $H_n=[x_{j,k}]$ with $x_{j,k}=e^{\frac{2\pi i}{n}jk}$ and $0\leq j,k\leq n-1$. One can verify that $e^ix_{j,k}$ is an imaginary number for any $j,k$. We construct the $n\times n$ diagonal unitary $
U=e^iI_d\oplus I_{n-d},
$ where $0\leq d\leq n$, as well as the 
$n\times n$ diagonal unitary
$
V=1\oplus e^{i}I_{n-1}.
$ One can verify that $UH_nV$ is an $n\times n$ CHM with  $n-d$ real entries, namely the $1$'s in the lower left corner of $UH_nV$. Since $0\leq d\leq n$, we obtain $\left\{0,1,\cdots,n\right\}\subseteq \cS_n$.

(ii) Consider the $n\times n$ CHM $H_n=[x_{j,k}]$ with $x_{j,k}=e^{\frac{2\pi i}{n}jk}$ and $0\leq j,k\leq n-1$. One can verify that $x_{j,k}$ is not a pure imaginary number. Further, $x_{j,k}$ is a real number if and only if $jk=0$. Let the $n\times n$ diagonal unitary $
U=I_{n-d}\oplus iI_d,
$ where $0\leq d\leq n-1$. So $UH_n$ has exactly $2n-d-1$ real elements. We have proven the assertion. 

(iii) Every $3\times3$ CHM can be written as $H_3=D_1 V D_2$ where $D_1$ and $D_2$ are both diagonal unitaries, and $V=\left(
\begin{array}{cccccc}
1 &  1 & 1 \\
1 &  \o & \o^2 \\
1 &  \o^2 & \o \\
\end{array}
\right)$ or $\left(
\begin{array}{cccccc}
1 &  1 & 1 \\
1 &  \o^2  & \o \\
1 &  \o & \o^2  \\
\end{array}
\right)$, where $\o=e^{2\p i\over3}$. Suppose $D_1=\diag(a_1,a_2,a_3)$, $D_2=\diag(b_1,b_2,b_3)$, where $a_i$, $b_i$ have modulus one. Let $V=\left(
\begin{array}{cccccc}
1 &  1 & 1 \\
1 &  \o & \o^2 \\
1 &  \o^2 & \o \\
\end{array}
\right)$. If we respectively choose $a_2=a_3=b_2=b_3=1$ and $D_2=I_3$ then $H_3$ becomes
\begin{eqnarray}
H_{31}&=&\bma a_1b_1 &a_1 &a_1\\b_1 &\omega &\omega^2\\ b_1 &\omega^2 &\omega\ema,
\quad
H_{32}=\bma a_1 &a_1 &a_1\\a_2 &a_2\omega &a_2\omega^2\\a_3 &a_3\omega^2 &a_3\omega\ema.
\end{eqnarray}
In $H_{31}$, let $(a_1,b_1)$ be $(1,1)$, $(1,\omega)$, $(\omega,\omega^2)$ and $(\omega,\omega)$, respectively. We respectively have $5,2,1,0\in\cS_3$. Further if $(a_1,b_1)=(\omega,\omega)$,
then $\o^2 H_{31}$ is a $3\times3$ CHM of six real entries. So $6\in\cS_3$. In $H_{32}$, let $(a_1,a_2,a_3)$ be $(1,i,i)$ and $(1,1,i)$, respectively. Then we have $3,4\in \cS_3$.

Evidently $H_3$ has no three real columns or three real rows, we have $9\notin \cS_3$. Using Lemma \ref{le:properties} (i), we obtain $7,8\notin \cS_3$. We have proven the assertion.

(iv) The $4\times4$ Hadamard matrix exists, say
\begin{eqnarray}
M=
\bma 
1&1&1&1\\
1&-1&1&-1\\
1&1&-1&-1\\
1&-1&-1&1\\
\ema.	
\end{eqnarray}
So $16\in \cS_4$. One can straightforwardly show that $0-10,12,16\in\cS_4$. For example by setting $D_1=\diag(1,i,i,i)$, we obtain that $D_1MD_1^\dg$ has $10$ real entries. By setting $D_2=\diag(i,1,1,1)$ and $D_3=\diag(1,i,i,e^{{\p i\over4}})$, we obtain that $D_2MD_3$ has $7$ real entries.

We prove that $11,13,14,15\not\in\cS_4$ by contradiction. Suppose $N$ is a $4\times4$ CHM containing exactly $11$ real entries. Lemma \ref{le:linearalg} (ii) shows that $N$ does not have real rows or columns. Up to equivalence 
we have 
$N=
\bma 
i_1 & * & * & *\\
* & i_2 & * & *\\
* & * & i_3 & i_4\\
* & * & * & i_5\\
\ema,
$
where $i_j$ is imaginary, and $*$ is a $1$ or $-1$. Column $1$ and $4$ of $N$ gives a contradiction with Lemma \ref{le:linearalg} (ii). So $11\not\in\cS_4$. One can similarly show that $13,14,15\not\in\cS_4$.

(v) The assertion follows from Proposition \ref{pp:0-2224252630inCS6} and \ref{pp:3133-36notinCS6} in Sec. \ref{sec:6x6proof}. In particular Proposition \ref{pp:3133-36notinCS6}  follows from Lemma \ref{le:23notincs6-a1=0}, \ref{le:23notincs6-a1=1}, and \ref{le:23notincs6-a1=2}.
\end{proof}

As an application of Theorem \ref{thm:s3456}, we present Theorem \ref{thm:app} as the second main result as follows.

\begin{theorem}
\label{thm:app}	
Any member of an MUB trio has at most $22$ real elements. 
\end{theorem}
\begin{proof}
Let $M_n$ be a member of an MUB trio having exactly $n$ real entries. It follows from Theorem \ref{thm:s3456} that $n\in\cS_6=\{0-22,24,25,26,30\}$. Suppose $n=25,26$ or $30$. So $M_n$ has two columns containing at most three imaginary entries. It implies that the two columns has a $3\times2$ real submatrix. So $M_{30}$ has been excluded as a member of any MUB trio by Lemma \ref{le:mubtrio}. Suppose $n=24$. Using the previous argument for $n=25,26$ or $30$, up to equivalence we may assume that  
column $2k-1,2k$ of $M_{24}=[m_{ij}],i,j=1,2,...,6$ has exactly four imaginary entries, for $k=1,2,3$ respectively. The previous argument we can assume that $m_{11},m_{21},m_{32},m_{42},m_{53},m_{63}$ are all imaginary entries. So there is an integer $j\in\{1,2,3\}$ such that column $j,4$ of $M_{24}$ has a $3\times2$ real submatrix. So $M_{24}$ is excluded by Lemma \ref{le:mubtrio} again. 
The assertion holds. 	
\end{proof}

\section{Proof of Theorem \ref{thm:s3456} (v)}
\label{sec:6x6proof}

We begin by characterizing the elements belonging to $\cS_6$.

\begin{proposition}
\label{pp:0-2224252630inCS6}
$0-22,24,25,26,30\in\cS_6$.
\end{proposition}
\begin{proof}
Consider the order-six CHM
\begin{eqnarray}
G_6=
\bma
i & 1 & 1 & 1 & 1 & 1\\
1 & i & -1 & -1 & 1 & 1\\
1 & -1 & i & 1 & 1 & -1\\
1 & -1 & 1 & i & -1 & 1\\
1 & 1 & 1 & -1 & i & -1\\
1 & 1 & -1 & 1 & -1 & i\\
\ema.	
\end{eqnarray}
One can show that $\cR(G_6)=30$. Using the matrices complex equivalent to $G_6$, one can construct matrices having $0-22,24-26,30$ real entries, respectively. The idea is as follows. If the first row of $G_6$ is multiplied by $i$ or $e^{{\p\over4}i}$ then the resulting matrix has $30-4=26$ or $30-5=25$ imaginary entries, respectively. We may repeat this argument by multiplying $i$ or $e^{{\p\over4}i}$ to row $2,3,4,5,6$ of $G_6$, respectively, and thus reduce the number of imaginary entries by $4$ or $5$. So we can construct CHMs containing exactly $m$ real entries with the flow of $m$ as follows.
\begin{eqnarray}
\label{eq:m=30->25-26}
&&
m=30\ra 25-26\ra20-22\ra15-18
\notag\\&&
\ra10-14\ra5-10\ra0-6.	
\end{eqnarray}
On the other hand, if we multiply the first row and second column of $M$ by $i$, respectively then the resulting CHM has exactly $24$ real entries. Further, if we multiply the first row, the second and third columns of $M$ by $i$, respectively, and the last row of $M$ by $e^{{\p\over4}i}$, then the resulting CHM has exactly $19$ real entries. Combining these results and \eqref{eq:m=30->25-26}, we obtain the assertion.
\end{proof}

Using Proposition \ref{pp:0-2224252630inCS6} for obtaining Theorem \ref{thm:app}, we need to show that the integers $23,27-29,31-36\not\in\cS_6$. For this purpose we define the \textit{imaginary array} of an order-$6$ CHM $M$ as follows. Let $a_i$ be the number of imaginary entries in the $i$'th row of $M$. Up to equivalence we may assume that 
\begin{eqnarray}
\label{eq:0lea1lea2}
0\le a_1\le a_2\le a_3\le a_4\le a_5\le a_6\le 6.	
\end{eqnarray}
Then the array $[a_1,a_2,a_3,a_4,a_5,a_6]$ is the imaginary array of $M$. Evidently, the sum of $a_i$'s is exactly the number of imaginary entries of $M$. Now we are in a position to show $23,27-29,31-36\not\in\cS_6$. 
\begin{proposition}
\label{pp:3133-36notinCS6}
(i) $31,33,34,35,36\not\in\cS_6$.	

(ii) $27,28,29,32\not\in\cS_6$.

(iii) $23\not\in\cS_6$.
\end{proposition}
\begin{proof}
(i) The assertion follows from Lemma \ref{le:properties} (ii.a) and the known fact that order-six real Hadamard matrix does not exist.

(ii) Suppose $32\in\cS_6$. Using Lemma \ref{le:properties} (ii.a), the CHM $M$ exists only if the four imaginary entries form a $2\times2$ submatrix of $M$. Lemma \ref{le:properties} (iv) implies that $M$ does not exist. So $32\not\in\cS_6$. In the following we assume that $a_i$ is the number of imaginary entries in the $i$'th row of $M$. If $M$ has exactly $29$ real entries then  the imaginary array $[a_1,a_2,a_3,a_4,a_5,a_6]=[1,1,1,1,1,2]$ up to equivalence by Lemma \ref{le:properties} (ii.a) and (iv). It is a contradiction with Lemma \ref{le:properties} (ii.b) and (vi). So $29\not\in\cS_6$.

Next we assume $28\in\cS_6$, and the CHM $M$ has exactly $28$ real entries. Lemma \ref{le:properties} (ii.a) implies that the imaginary array $[a_1,a_2,a_3,a_4,a_5,a_6]=[1,1,1,1,2,2]$ or $[1,1,1,1,1,3]$. Lemma \ref{le:properties} (vi) shows that the imaginary entries are all $i$ or $-i$, and the imaginary entries in the first four or five rows of $M$ are in distinct columns. The former is excluded by Lemma \ref{le:properties} (ii.b). The latter is excluded by the non-orthogonality of a column vector containing two imaginary entries, and the last column vector of $M$. 
So $28\not\in\cS_6$.

Finally we assume $27\in\cS_6$, and the CHM $M$ has exactly $27$ real entries. Lemma \ref{le:properties} (ii.b) shows that $a_j\ge1$ for any $j$. So the imaginary array of $M$ is $[a_1,a_2,a_3,a_4,a_5,a_6]=[1,1,1,2,2,2]$,  $[1,1,1,1,2,3]$ or $[1,1,1,1,1,4]$. They are all excluded by Lemma \ref{le:properties} (ii.b), since in each case $M$ has two rows or two columns having one and two imaginary entries, respectively.

(iii) We prove the assertion by contradiction. Suppose $M$ is a $6\times6$ CHM having exactly $23$ real entries. That is, $M$ has exactly $13$ imaginary entries. So the imaginary array $[a_1,a_2,a_3,a_4,a_5,a_6]$ of $M$ satisfies
\begin{eqnarray}
\label{eq:sumai=13}	
\sum^6_{i=1}a_i=13. 
\end{eqnarray}
Hence we have three subcases, namely $a_1=0,1$, or $2$. In either case we show that $M$ does not exist. We shall provide their proofs in the subsequent Lemma \ref{le:23notincs6-a1=0}, \ref{le:23notincs6-a1=1}, and \ref{le:23notincs6-a1=2}, respectively.  
\end{proof}

\begin{lemma}
\label{le:23notincs6-a1=0}  	
The $6\times6$ CHM containing exactly 23 real entries does not exist, if it has a row containing no imaginary entry. 
\end{lemma}
\begin{proof}
We shall follow the notation in the proof of Proposition \ref{pp:3133-36notinCS6} (iii). We have $a_1=0$, namely the first row of $M$ has no imaginary entries. We have $a_i\ne 1$ for $i=2,3,4,5,6$ by Lemma \ref{le:properties} (ii). Next we have $a_i\ne 3$ by Lemma \ref{le:linearalg} (ii). Eqs. \eqref{eq:sumai=13} and \eqref{eq:0lea1lea2} imply that some $a_i=5$. If $a_6=6$ then \eqref{eq:sumai=13} and \eqref{eq:0lea1lea2} imply that $a_2=a_3=0$ and $a_4=2$. So $M$ has three real rows, and it is a contradiction with Lemma \ref{le:properties} (iv). Hence $a_6<6$, and we have $a_6=5$. Eqs. \eqref{eq:sumai=13} and \eqref{eq:0lea1lea2} imply that $a_i=0,2,4$ for $i<6$. So there are three cases, namely $(a_1,a_2,a_3,a_4,a_5)=(0,0,0,4,4),(0,0,2,2,4)$ and $(0,2,2,2,2)$. The first case is excluded by Lemma \ref{le:properties} (iv). The second case implies that the first three rows of $M$ are 
$\bma 
1&1&1&1&1&1\\
1&1&1&-1&-1&-1\\
a&b&c&d&e&f\\
\ema$. Since they are pairwise orthogonal, Lemma \ref{le:linearalg} (i) implies that two of $a,b,c$ are imaginary, and two of $d,e,f$ are also imaginary. So $a_3\ge4$, and it is a contradiction with the fact that $a_3=2$. We have excluded the second case. 

It remains to investigate the third case, namely the order-$6$ CHM $M$ has the imaginary array
\begin{eqnarray}
\label{eq:022225}	
[a_1,a_2,a_3,a_4,a_5,a_6]=[0,2,2,2,2,5].
\end{eqnarray}
Using Lemma \ref{le:properties} (ii.d), we may assume that the first three rows of $M$ form one of the four matrices $H_{61}-H_{64}$ in \eqref{eq:row123-1}-\eqref{eq:row123-4}. Since $H_{61}=H_{62}^*$ and $H_{63}=H_{64}^*$, it suffices to show that the first three rows of $M$ cannot be $H_{61}$ in \eqref{eq:row123-1} or $H_{63}$ in \eqref{eq:row123-3}. We prove it by contradiction. Using \eqref{eq:022225} and equivalence, we can permute the row $2-5$ of $m$ so that they still have exactly two imaginary entries. Applying the pigeonhole principle to row $2-5$ of $M$, one can show that the first three rows of $M$ form the matrix $H_{61}$ namely
\begin{eqnarray}
\label{eq:row123-11}
&&
H_{61}=
\bma
1&1& 1&1&1&1 \\
i&-i& 1&1&-1&-1\\
i& 1 & -i& -1& -1& 1\\
\ema.
\end{eqnarray}
Using Lemma \ref{le:properties} (ii.d) and \eqref{eq:022225}, we obtain that row $4$ and $5$ of $M$ both have exactly two imaginary entries, and exactly one of them is in the first three columns of $M$. Lemma \ref{le:properties} (ii.b) implies that the entries are $i,-i$. Using \eqref{eq:row123-11}, the first four rows of $M$ is one of the following two matrices.
\begin{eqnarray}
\label{eq:row123-11}
&&
H_{611}=
\bma
1&1& 1&1&1&1 \\
i&-i& 1&1&-1&-1\\
i& 1 & -i& -1& -1& 1\\
i& a_1 & b_1 & -i& c_1& d_1\\
\ema,
\\&&
H_{612}=
\bma
1&1& 1&1&1&1 \\
i&-i& 1&1&-1&-1\\
i& 1 & -i& -1& -1& 1\\
a_2 & i & -i & b_2 & c_2& d_2\\
\ema,
\end{eqnarray}
where $\{a_j,b_j,c_j,d_j\}=\{1,1,-1,-1\}$ for $j=1,2$. Using the same argument, only $H_{611}$ may have row $5$ containing two imaginary elements, and orthogonal to row $1-4$. We may assume that the first five rows of $M$ are 
\begin{eqnarray}
H_{6111}=
\bma
1&1& 1&1&1&1 \\
i&-i& 1&1&-1&-1\\
i& 1 & -i& -1& -1& 1\\
i& a_1 & b_1 & -i& c_1& d_1\\
i& a_3 & b_3 & c_3& -i & d_3\\
\ema.
\end{eqnarray}
The orthogonality between row $2,3,4$ implies $(a_1,b_1,c_1,d_1)=(1,-1,1,-1)$. Similarly, $(a_3,b_3,c_3,d_3)=(-1,-1,1,1)$. Hence
\begin{eqnarray}
H_{6111}=
\bma
1&1& 1&1&1&1 \\
i&-i& 1&1&-1&-1\\
i& 1 & -i& -1& -1& 1\\
i& 1 & -1 & -i& 1& -1\\
i& -1 & -1 & 1& -i & 1\\
\ema.
\end{eqnarray}
Since the first five row vectors of $M$ are orthogonal to the last row, using Lemma \ref{le:properties} (ii.b) we have
\begin{eqnarray}
M=
\bma
1&1& 1&1&1&1 \\
i&-i& 1&1&-1&-1\\
i& 1 & -i& -1& -1& 1\\
i& 1 & -1 & -i& 1& -1\\
i& -1 & -1 & 1& -i & 1\\
u & x & -x & x & -x & -u\\
\ema.
\end{eqnarray}
So the last row of $M$ does not have exactly five imaginary entries. It is a contradiction with $a_6=5$. We have shown that $a_1=0$ is impossible. 
\end{proof}

\begin{lemma}
\label{le:23notincs6-a1=1}  	
The $6\times6$ CHM containing exactly 23 real entries does not exist, if it has a row containing exactly one imaginary entry. 
\end{lemma}
\begin{proof}
We shall follow the notation in the proof of Proposition \ref{pp:3133-36notinCS6} (iii). Since the $6\times6$ CHM $M$ has a row containing exactly one imaginary entry, we have $a_1\le 1$. If $a_1=0$ then the assertion follows from Lemma \ref{le:23notincs6-a1=0}. We have $a_1=1$, namely the first row of $M$ has exactly one imaginary entry. If the set $\{a_2,a_3,a_4,a_5,a_6\}$ has one $1$ and one $2$, then Lemma \ref{le:properties} (v) implies that $M$ has a submatrix $\bma 
r_1&r_2&r_3&r_4&r_5&i_1\\
r_6&r_7&r_8&i_2&i_3&r_9\\
\ema$ with real $r_j$ and $i_k$. It is a contradiction with Lemma \ref{le:properties} (ii.b). So the imaginary array $[a_1,a_2,a_3,a_4,a_5,a_6]$ of $M$ has following five cases (i)-(v).
\begin{eqnarray}
&&
[1,1,1,1,3,6],
\quad
[1,1,1,1,4,5], 
\quad
[1,2,2,2,2,4],
\notag\\&&
[1,2,2,2,3,3],
\quad
[1,1,1,3,3,4].
\label{eq:1row=1imaginary}
\end{eqnarray}
In the following we shall show that $M$ does not exist in either of the five cases, respectively. It proves the assertion. 

(i) $[a_1,a_2,a_3,a_4,a_5,a_6]=[1,1,1,1,3,6].$ Lemma \ref{le:properties} (vi) implies that the four imaginary entries are distributed in different rows and columns of $M$. So
\begin{eqnarray}
\label{eq:m=r1r2}
M=\bma
r_1 & r_2 & r_3 & r_4 & r_5 & i\\	
r_6 & r_7 & r_8 & r_9 & i & r_{10} \\	
r_{11} & r_{12} & r_{13} & i & r_{14} & r_{15} \\
r_{16} & r_{17} & i & r_{18} & r_{19} & r_{20} \\
y_1 & y_2 & y_3 & y_4 & y_5 & y_6 \\
x_1 & x_2 & x_3 & x_4 & x_5 & x_6 \\
\ema,
\end{eqnarray}
where $r_j$ is real and $x_k$ is imaginary. If $y_1,y_2$ are both real, then multiplying $x_1^*$ on the bottom row of $M$ implies a contradiction with Lemma \ref{le:properties} (ii.a) and (ii.b). Next if $y_1,y_2$ are both imaginary then it is a contradiction with Lemma \ref{le:properties} (vi).
So the only possibility is that exactly one of $y_1,y_2$ is imaginary. Since $a_5=3$, exactly two of $y_3,y_4,y_5,y_6$ are imaginary. Using equivalence on $M$, we can assume that $y_2,y_4,y_5$ are imaginary, and $y_1,y_3,y_6$ are real. Using \eqref{eq:m=r1r2} and Lemma \ref{le:properties} (ii.b), we obtain that $x_1^*x_4$ and $x_1^*x_5$ are both real. Now column $1,5,6$ of $M$ is a contradiction with Lemma \ref{le:properties} (ii.d).  We have excluded case (i). 

(ii) $[a_1,a_2,a_3,a_4,a_5,a_6]=[1,1,1,1,4,5].$ Using the argument for (i), one can obtain $M$ in \eqref{eq:m=r1r2} such that $r_j,y_1,y_3$ are real, $y_2,y_4,y_5,y_6$ are imaginary, and five of $x_1,x_2,..,x_6$ are imaginary. In particular two of $x_4,x_5,x_6$ are imaginary. So column $1,4,5,6$ of $M$ and Lemma \ref{le:properties} (ii.b) imply that $x_1$ is imaginary. The same argument implies that $x_4,x_5,x_6$ are all imaginary. So $x_2,x_3$ are real and imaginary, respectively. Column $2$ and $3$ give a contradiction with Lemma \ref{le:properties} (ii.b).

(iii) $[a_1,a_2,a_3,a_4,a_5,a_6]=[1,2,2,2,2,4].$ It follows from Lemma \ref{le:properties} (vi) that a column of $M$ has five imaginary entries. This case has been excluded by the previous cases because the transpose of a CHM is still a CHM.

(iv) $[a_1,a_2,a_3,a_4,a_5,a_6]=[1,2,2,2,3,3].$ Up to equivalence, we may assume that the imaginary entry in the first row of $M$ is the last entry of the first row. Since $a_2=a_3=a_4=2$, Lemma \ref{le:properties} (ii.b) shows that the first four entries of the last row of $M$ are all imaginary. Lemma \ref{le:properties} (v) shows that the three imaginary entries not in the last column of $M$ are in distinct rows and columns of $M$. Up to equivalence, the above argument shows 
\begin{eqnarray}
\label{eq:m=r1r2}
M=\bma
r_1 & r_2 & r_3 & r_4 & r_5 & w_1\\	
r_6 & r_7 & r_8 & r_9 & w_2 & w_6 \\	
r_{10} & r_{11} & r_{12} & w_3 & r_{13} & w_7 \\
r_{14} & r_{15} & w_4 & r_{16} & r_{17} & w_8 \\
x_1 & x_2 & x_3 & x_4 & x_5 & x_6 \\
y_1 & y_2 & y_3 & y_4 & y_5 & y_6 \\
\ema,
\end{eqnarray}
where $r_i$ is real, $w_j$ is imaginary, $x_k,y_k$ are complex. Since $a_5=a_6=3$, $\{x_k\}$ and $\{y_k\}$ has exactly three imaginary entries, respectively. If $x_6$ or $y_6$ is imaginary then the column of $M$ has at least five imaginary entries. The case has been excluded by previous cases in this lemma using the transpose of $M$. So 
$x_6$ and $y_6$ are both real. If $x_1,x_2,y_1,y_2$ are real, then $a_5=a_6=3$ implies that $x_i,y_i$ are imaginary for $i=3,4,5$. Column $1$ and $3$ give a contradiction with Lemma \ref{le:properties} (ii.b). So one of $x_1,x_2,y_1,y_2$ is imaginary. Up to equivalence we may assume that $x_2$ is imaginary. Since $a_5=3$, one of $x_3,x_4,x_5$ is imaginary. Up to equivalence we may assume that $x_3$ is imaginary. Using \eqref{eq:m=r1r2} we summary the above findings as follows. 
\begin{eqnarray}
\label{eq:m=r1r2-row56}
M=\bma
r_1 & r_2 & r_3 & r_4 & r_5 & w_1\\	
r_6 & r_7 & r_8 & r_9 & w_2 & w_6 \\	
r_{10} & r_{11} & r_{12} & w_3 & r_{13} & w_7 \\
r_{14} & r_{15} & w_4 & r_{16} & r_{17} & w_8 \\
x_1 & w_5 & w_9 & x_4 & x_5 & r_{18} \\
y_1 & y_2 & y_3 & y_4 & y_5 & r_{19} \\
\ema,
\end{eqnarray}
where $r_i$ is real, $w_j$ is imaginary and $x_k,y_l$ are complex. Since $a_5=3$, exactly one of $x_1,x_4,x_5$ is imaginary. Up to the equivalence we have two cases (iv.a) and (iv.b), namely $x_1$ or $x_4$ is imaginary. 

(iv.a) $x_1$ in \eqref{eq:m=r1r2-row56} is imaginary. 
Lemma \ref{le:properties} (vii) shows that one of $y_1,y_2$ is imaginary. By permuting column $1$ and $2$ of $M$ we may assume that $y_2$ is imaginary. Since column $2,4,5$ of $M$ are pairwise orthogonal, Lemma \ref{le:properties} (ii.b) shows that $y_4,y_5$ are imaginary. Column $1,4$ of $M$ are pairwise orthogonal. It is a contradiction with $a_6=3$ and Lemma \ref{le:properties} (ii.b).

(iv.b) $x_4$ in \eqref{eq:m=r1r2-row56} is imaginary. So $x_1$ is real. If $y_1$ is real then the case has been excluded by Lemma \ref{le:23notincs6-a1=0}.
We obtain that $y_1$ is imaginary. Since column $1,3,4$ of $M$ are pairwise orthogonal, Lemma \ref{le:properties} (ii.b) shows that $y_3,y_4$ are imaginary, and thus $y_1^*y_3$ and $y_1^*y_4$ are both real. So the transpose of $\diag(1,1,1,1,1,y_1^*)M$ has $(a_1,a_2,a_3,a_4,a_5,a_6)=(0,2,2,2,2,5)$. It has been excluded by Lemma \ref{le:23notincs6-a1=0}.

(v) 
\begin{eqnarray}
\label{eq:111334}
[a_1,a_2,a_3,a_4,a_5,a_6]=[1,1,1,3,3,4].	
\end{eqnarray}
Lemma \ref{le:properties} (vi) implies that up to equivalence, the imaginary entries in the first three rows of $M$ are the first three diagonal entries of $M=[m_{ij}]_{i,j=1,...,6}$. That is, $m_{11},m_{22},m_{33}$ are imaginary. Let $(a,b)$ be the array of numbers of imaginary entries in the two rows of the following $2\times3$ submatrix, respectively.  
\begin{eqnarray}
\label{eq:111334=3x2}	
\bma m_{44}& m_{45}& m_{46}\\m_{54}& m_{55}& m_{56} \ema.
\end{eqnarray}
Let $a\le b$ up to the permutation of row $4$ and $5$ of $M$. If $a=0$ then $M$ has a $4\times3$ real submatrix. It is a contradiction with Lemma \ref{le:properties} (v). So $(a,b)$ has the following six subcases (v.a)-(v.f).
\begin{eqnarray}
\label{eq:111334=6cases}
(1,1),(1,2),(1,3),(2,2),(2,3),(3,3).	
\end{eqnarray}  
In the following we shall investigate the six subcases, respectively. It proves assertion (v), and thus the last case in \eqref{eq:1row=1imaginary}.

(v.a) $(a,b)=(1,1)$ in \eqref{eq:111334=3x2}	 and \eqref{eq:111334=6cases}. Using \eqref{eq:111334} and \eqref{eq:111334=3x2}, up to equivalence we may assume that $m_{42},m_{43},m_{44}$ are imaginary. If $m_{54}$ is imaginary then $b=1$ implies that $m_{55},m_{56}$ are both real. The right most three columns of $M$ gives a contradiction with Lemma \ref{le:properties} (vii). So $m_{54}$ is real. Since $b=1$, \eqref{eq:111334=3x2} implies that one of $m_{55},m_{56}$ is imaginary. Up to the permuting of column $5$ and $6$ of $M$, we may assume that $m_{55}$ is imaginary.
Since $m_{i6}$ is real for $i=1,...,5$, the case of real $m_{66}$ has been excluded by Lemma \ref{le:23notincs6-a1=0}. So $m_{66}$ is imaginary. Concluding the above findings, we write $M$ by marking the imaginary entries as $w_j$'s as follows.
\begin{eqnarray}
\label{eq:(a,b)=(1,1)}
M=
\bma 
w_1 & m_{12} & m_{13} & m_{14} & m_{15} & m_{16}\\
m_{21} & w_2 & m_{23} & m_{24} & m_{25} & m_{26}\\
m_{31} & m_{32} & w_3 & m_{34} & m_{35} & m_{36}\\
m_{41} & w_4 & w_5 & w_6 & m_{45} & m_{46}\\
m_{51} & m_{52} & m_{53} & m_{54} & w_7 & m_{56}\\
m_{61} & m_{62} & m_{63} & m_{64} & m_{65} & w_8\\
\ema. 	
\end{eqnarray}
Lemma \ref{le:properties} (vi) implies that $w_1,w_2,w_3$ are all $i$ or $-i$. 
Lemma \ref{le:properties} (vi) and column $1,2,3,4$ of $M$ imply that $w_4,w_5$ are $i$ or $-i$, and so is $w_6$. Recall that $(a,b)=(1,1)$ in \eqref{eq:111334=3x2}. Thus $m_{54},m_{56}$ in \eqref{eq:(a,b)=(1,1)} are both real. If $m_{64}$ in \eqref{eq:(a,b)=(1,1)} is real, then multiplying row $4$ of $M$ by $i$ implies a CHM of exactly $13$ imaginary entries, whose column $4$ is real. Such $M$ has been excluded by Lemma \ref{le:23notincs6-a1=0}. So $m_{64}$ in \eqref{eq:(a,b)=(1,1)} is imaginary. Lemma \ref{le:properties} (ii.b) and column $4,5$ of $M$ imply that $m_{65}$ in \eqref{eq:(a,b)=(1,1)} is also imaginary. Multiplying row $4$ of $M$ by $i$ implies a CHM of exactly $13$ imaginary entries, Since $(a,b)=(1,1)$ in \eqref{eq:111334=3x2}, the CHM has at least one and at most two columns each of which contains exactly one imaginary entry. Such $M$ has been excluded by case (i)-(iv) of this lemma. We have excluded case (v.a).

(v.b) $(a,b)=(1,2)$ in \eqref{eq:111334=3x2}	 and \eqref{eq:111334=6cases}. Using Lemma \ref{le:properties} (vii), we may assume that $m_{42},m_{43},m_{44},m_{54},m_{55}$ are imaginary. Lemma \ref{le:23notincs6-a1=0} shows that $m_{66}$ is imaginary. 
Applying Lemma \ref{le:properties} (ii.b) to column $4,6$ of $M$, we obtain that $m_{64}$ is imaginary. Lemma \ref{le:properties} (ii.b) and column $2,3,6$ show that $m_{62},m_{63}$ are imaginary. Since $(a,b)=(1,3)$ in \eqref{eq:111334=3x2}	 and \eqref{eq:111334=6cases}, case (i)-(iv) of this lemma shows that $m_{52}$ is imaginary. 

We claim that all imaginary entries of $M$ are $i$ or $-i$. First Lemma \ref{le:properties} (vi) implies that $w_1,w_2,w_3$ are all $i$ or $-i$. Next Lemma \ref{le:properties} (ii.b) and the orthogonality of rows and columns of $M$ show that the remaining imaginary entries of $M$ are all $i$ or $-i$. So row $3,6$ contradicts with Lemma \ref{le:properties} (ii.b).

(v.c) $(a,b)=(1,3)$ in \eqref{eq:111334=3x2}	 and \eqref{eq:111334=6cases}. Up to equivalence we may assume that $m_{42},m_{43},m_{44.},m_{54},m_{55},m_{56}$ are imaginary. So column $1,5,6$ of $M$ gives a contradiction with the first type in Lemma \ref{le:properties} (vii). 

(v.d) $(a,b)=(2,2)$ in \eqref{eq:111334=3x2}	 and \eqref{eq:111334=6cases}. Similar to the arguments in (v.a)-(v.c), one can obtain that $m_{11},m_{22},m_{33},m_{43},m_{44},m_{45}$ are imaginary, and the first four of them are $i$ or $-i$. Since $b=2$, we investigate the position of imaginary entries in $m_{54},m_{55},m_{56}$. There are two cases (v.d.1) and (v.d.2) as follows. (v.d.1) If $m_{54},m_{55}$ are both imaginary, then Lemma \ref{le:23notincs6-a1=0} implies that $m_{66}$ is imaginary. Since column vector $4,5$ of $M$ are both orthogonal to column vector $6$, Lemma \ref{le:properties} (ii.b) implies that $m_{64},m_{65}$ are equal to $m_{66}$ or $-m_{66}$. So the product matrix $\diag(1,1,1,1,1,m_{66}^*)M$ has a $4\times3$ real submatrix. It is a contradiction with Lemma \ref{le:properties} (v). (v.d.2) On the other hand if $m_{55},m_{56}$ are both imaginary, then the only case not excluded by Lemma \ref{le:23notincs6-a1=0} and case (i)-(iv) of this lemma occurs when the imaginary array of $M^T$ is $[2,2,2,2,2,3]$. If $m_{53}$ is imaginary then $m_{63}$ is real. So the product matrix $M\diag(1,1,i,1,1,1)$ has a real third column, and it still has exactly $13$ imaginary entries. This case has been excluded by Lemma \ref{le:23notincs6-a1=0}. On the other hand if $m_{53}$ is real, up to the permuting of column $1,2$ and row $1,2$ of $M$ we may assume that $m_{52}$ is imaginary. Since $a_6=4$, we obtain that $m_{61},m_{64},m_{66}$ are imaginary, and one of $m_{62},m_{63}$ is imaginary. Using Lemma \ref{le:properties} (ii.b) we can obtain that the imaginary entries in $M$ are all $i$ or $-i$. So row $1,6$ of $M$ are not orthogonal, and we have a contradiction. 

(v.e) $(a,b)=(2,3)$ in \eqref{eq:111334=3x2}	 and \eqref{eq:111334=6cases}. Similar to the arguments in (v.a)-(v.d), one can obtain that
\begin{eqnarray}
\label{eq:(a,b)=(2,3)}
M=
\bma 
w_1 & m_{12} & m_{13} & m_{14} & m_{15} & m_{16}\\
m_{21} & w_2 & m_{23} & m_{24} & m_{25} & m_{26}\\
m_{31} & m_{32} & w_3 & m_{34} & m_{35} & m_{36}\\
m_{41} & m_{42}  & w_4 & w_5 & w_6 & m_{46}\\
m_{51} & m_{52} & m_{53} & w_7 & w_8 & w_9\\
m_{61} & m_{62} & m_{63} & m_{64} & m_{65} & m_{66}\\
\ema. 	
\end{eqnarray}
Since $a_4=4$ in \eqref{eq:111334}, any column of $M$ in \eqref{eq:(a,b)=(2,3)} has at most three imaginary entries. So the only case not excluded by Lemma \ref{le:23notincs6-a1=0} and case (i)-(iv) of this lemma occurs when the imaginary array of $M^T$ is $[2,2,2,2,2,3]$. So $m_{61},m_{62}$ and $m_{66}$ are imaginary, and one of $m_{63},m_{64}$ is imaginary up to the permuting of column $4,5$ of $M$. Using Lemma \ref{le:properties} (ii.b) and (vi), one can show that the imaginary entries in $M$ are $i$ or $-i$. So row $1,6$ of $M$ are not orthogonal. 

(v.f) $(a,b)=(3,3)$ in \eqref{eq:111334=3x2}	 and \eqref{eq:111334=6cases}. From the previous argument, we obtain that the entries $m_{11},m_{22},m_{33},m_{44},m_{45},m_{46},m_{54},m_{55},m_{56}$ are imaginary. So the only case not excluded by Lemma \ref{le:23notincs6-a1=0} and case (i)-(iv) of this lemma occurs when the imaginary array of $M$ is $[2,2,2,2,2,3]$. Up to the permuting of columns of $M$, using $a_6=4$ we may assume that $m_{61},..,m_{64}$ are imaginary. Since column $4,5$ of $M$ are orthogonal, up to the permuting of row $4,5$ of $M$ we may assume that $w_4=w_5$ or $-w_5$ by Lemma \ref{le:properties} (ii.b). Similarly, since column $4,6$ of $M$ are orthogonal we obtain that $w_7=w_9$ or $-w_9$ by Lemma \ref{le:properties} (vii). The orthogonality of column $4,5,6$ implies
\begin{eqnarray}
&& w_4=pw_5,
\quad\quad
w_7=qw_9,
\\&&	
\label{eq:w5w6c+w8w9c}
w_5w_6^*+w_8w_9^*=0,
\\&&
w_4w_5^*+w_7w_8^*+w_{13}=0,
\\&&
w_4w_6^*+w_7w_9^*+w_{13}=0.
\end{eqnarray}
where $p,q=\pm1$. Lemma \ref{le:linearalg} (i) implies that $(pqw_9w_8^*,pw_{13})=(\o,\o^2)$ or $(\o^2,\o)$, and $(pqw_5w_6^*,qw_{13})=(\o,\o^2)$ or $(\o^2,\o)$. So $p=q$, and \eqref{eq:w5w6c+w8w9c} is not satisfied. We have a contradiction.
\end{proof}

\begin{lemma}
\label{le:23notincs6-a1=2}  	
The $6\times6$ CHM containing exactly 23 real entries does not exist, if it has a row containing exactly two imaginary entries. 
\end{lemma}
\begin{proof}
We shall follow the notation in the proof of Proposition \ref{pp:3133-36notinCS6} (iii). Since the $6\times6$ CHM $M$ has a row containing exactly two imaginary entry, we have $a_1\le 2$. If $a_1\le1$ then the assertion follows from Lemma \ref{le:23notincs6-a1=1}. We have $a_1=2$, namely the first row of $M$ has exactly two imaginary entries. Hence, the imaginary array of $M$ is 
\begin{eqnarray}
\label{eq:222223}	
[a_1,a_2,a_3,a_4,a_5,a_6]=[2,2,2,2,2,3].
\end{eqnarray}
Using Lemma \ref{le:23notincs6-a1=0} and \ref{le:23notincs6-a1=1}, we may assume that the imaginary array of $M^T$ is also $[2,2,2,2,2,3]$. We have two cases in terms of the row and column of $M$ containing three imaginary entries. The first case is that the row and column do not have common imaginary entry, and the second case is that they do. Up to equivalence we may assume that in the two cases $M$ become respectively
\begin{eqnarray}
\label{eq:222223=case1}
\bma 
m_{11} & w_1 & w_2 & w_3 & m_{15} & m_{16}\\
w_4 & m_{22} & m_{23} & m_{24} & m_{25} & m_{26}\\
w_5 & m_{32} & m_{33} & m_{34} & m_{35} & m_{36}\\
w_6 & m_{42} & m_{43} & m_{44} & m_{45} & m_{46}\\
m_{51} & m_{52} & m_{53} & m_{54} & m_{55} & m_{56}\\
m_{61} & m_{62} & m_{63} & m_{64} & m_{65} & m_{66}\\
\ema,
\end{eqnarray}
and
\begin{eqnarray}
\label{eq:222223=case2}
\bma 
w_1 & w_2 & w_3 & m_{14} & m_{15} & m_{16}\\
w_4 & m_{22} & m_{23} & m_{24} & m_{25} & m_{26}\\
w_5 & m_{32} & m_{33} & m_{34} & m_{35} & m_{36}\\
m_{41} & m_{42} & m_{43} & m_{44} & m_{45} & m_{46}\\
m_{51} & m_{52} & m_{53} & m_{54} & m_{55} & m_{56}\\
m_{61} & m_{62} & m_{63} & m_{64} & m_{65} & m_{66}\\
\ema,
\end{eqnarray}
where $w_1,w_2,...,w_6$ are imaginary entries. Since \eqref{eq:222223}	holds for both rows and columns of $M$, row $2,3,4$ of \eqref{eq:222223=case1} all have exactly two imaginary entries. Lemma \ref{le:properties} (ii.b) implies that $w_4w_5^*$ and $w_4w_6^*$ are $1$ or $-1$. Multiplying the first column of \eqref{eq:222223=case1} results in a CHM whose imaginary array is not $[2,2,2,2,2,3]$. It has been excluded by Lemma \ref{le:23notincs6-a1=0} and \ref{le:23notincs6-a1=1}. In the following we investigate $M$ in \eqref{eq:222223=case2}.
Recall that such $M$ and $M^T$ have the imaginary array $[2,2,2,2,2,3]$. So the submatrix $\bma m_{22}& m_{23}\\ m_{32}& m_{33}
\ema$ in \eqref{eq:222223=case2} has exactly $2,1$ or $0$ imaginary entry.
We shall investigate them respectively in case (i), (ii) and (iii). It turns out that in either case, $M$ in \eqref{eq:222223=case2} does not exist. So the assertion holds.

(i) The submatrix $\bma m_{22}& m_{23}\\ m_{32}& m_{33}
\ema$ in \eqref{eq:222223=case2} has exactly $2$ imaginary entries. Up to permuting of row and column $2,3$ we may assume that $m_{22},m_{33}$ are both imaginary. Applying Lemma \ref{le:properties} (viii) to row $2,3$ and column $2,3$ of $M$, we obtain that $w_2,w_3$ are equal or opposite, and so are $w_4,w_5$. Hence $w_2,w_3,...,w_7$ are pairwise equal or opposite. Since row vector $1,2$ of $M$ are orthogonal, we obtain that $w_1$ is real. It is a contradiction with the fact that $w_1$ is imaginary. Case (i) has been excluded. 

(ii) The submatrix $\bma m_{22}& m_{23}\\ m_{32}& m_{33}
\ema$ in \eqref{eq:222223=case2} has exactly $1$ imaginary entry. Up to equivalence we may assume that $m_{22}$ and $m_{34}$ in \eqref{eq:222223=case2} are imaginary. We denote $m_{22}=w_6$ and $m_{34}=w_7$. Similar to (i), one can obtain that $w_2,w_3,w_4,w_5,w_6$ are pairwise equal or opposite. 
Up to equivalence we may assume that $w_2=w_3=w_4=w_5=w$. Since row vector $1,2$ of $M$ in \eqref{eq:222223=case2} are orthogonal, we have 
\begin{eqnarray}
\label{eq:m223w}
m_{23}w^*+ww_1^*=w_6w^*+m_{24}+m_{25}+m_{26}=0	
\end{eqnarray}
by Lemma \ref{le:properties} (ii.c). Since column vector $1,2$ of $M$ in \eqref{eq:222223=case2} are orthogonal, we have $m_{32}w^*+ww_1^*=0$. The above equations imply $m_{32}=m_{23}$. Next the orthogonality of row $1,3$ of $M$ in \eqref{eq:222223=case2} implies $m_{33}w^*+w_7+m_{35}+m_{36}=0$. Since $w,w_7$ are imaginary, we have $m_{33}w^*+w_7=m_{35}+m_{36}=0$. Now using the orthogonality of row $2,3$ of $M$ in \eqref{eq:222223=case2} we have $m_{32}w_6^*=m_{24}m_{33}w^*$, namely $w_6^*=\pm w^*$. In either case \eqref{eq:m223w} and subsequent equations imply a contradiction. 
Case (ii) has been excluded. 

(iii) The submatrix $\bma m_{22}& m_{23}\\ m_{32}& m_{33}
\ema$ in \eqref{eq:222223=case2} has no imaginary entry. We have two cases, namely whether the remaining two imaginary entries in row $2,3$ of $M$ in \eqref{eq:222223=case2} are in the same column. For the first case, suppose they are in the same column. For convenience we describe the upper-left $3\times4$ submatrix of $M$ as
\begin{eqnarray}
\bma
w_1 & w_2 & w_3 & m_{14}\\
w_4 & m_{22} & m_{23} & w_6\\
w_5 & m_{32} & m_{33} & w_7\\
\ema,	
\end{eqnarray} 
where $w_j$'s are imaginary. The row vector $2,3$ are orthogonal, and the column vector $1,4$ are also orthogonal. So we have
\begin{eqnarray}
\label{eq:w4w5*}
&& w_4w_5^*+w_6w_7^*=0
\text{\quad or \quad} \pm2,
\\&&	
\label{eq:w4w6*}
w_4w_6^*+w_5w_7^*+w_1=\pm 1.
\end{eqnarray}
Since $w_1$ is imaginary, \eqref{eq:w4w5*} has to be $\pm2$. So $w_4w_5^*$ and $w_6w_7^*$ are both real numbers, namely $1$ or $-1$. It implies that the product of $w_4w_6^*$ and $w_5w_7^*$ is real. So they are both real or both imaginary. It is a contradiction with \eqref{eq:w4w6*} in terms of Lemma \ref{le:linearalg} (ii). 

It remains to investigate the second case of (iii), namely the remaining two imaginary entries in row $2,3$ of $M$ in \eqref{eq:222223=case2} are not in the same column. Applying the same argument to column $2,3$ of $M$, we obtain that the two imaginary entries in column $2,3$ are not in the same row. Up to equivalence we may assume that 
\begin{eqnarray}
\label{eq:222223=case2-iii}
M=
\bma 
w_1 & w_2 & w_3 & m_{14} & m_{15} & m_{16}\\
w_4 & m_{22} & m_{23} & w_6 & m_{25} & m_{26}\\
w_5 & m_{32} & m_{33} & m_{34} & w_7 & m_{36}\\
m_{41} & w_8 & m_{43} & m_{44} & m_{45} & m_{46}\\
m_{51} & m_{52} & w_9 & m_{54} & m_{55} & m_{56}\\
m_{61} & m_{62} & m_{63} & m_{64} & m_{65} & m_{66}\\
\ema,
\end{eqnarray}
where $w_j$'s are imaginary. Since column $2,3$ of $M$ in \eqref{eq:222223=case2-iii} are orthogonal, Lemma \ref{le:properties} (ii.b) shows that $w_2,w_3$ are equal or opposite. Similar arguments show that so are $w_4,w_5$, so are $w_6,w_7^*$ and so are $w_8,w_9^*$. Up to equivalence we may assume that $w_2=w_3$, $w_4=w_5$, $w_6=w_7^*$, and $w_8=w_9^*$. Since row $2,3$ of $M$ are orthogonal, we have $m_{34}=-m_{25}$, and the submatrix $N:=\bma m_{22} & m_{23}\\ m_{32}&m_{33} \ema$ has at least one $1$ and one $-1$. Similarly we have $m_{43}=-m_{52}$. Up to the multiplication of $-1$ on $M$, we may assume that the submatrix $N$ has three $1$'s and one $-1$, or two $1$'s and two $-1$'s. There are three cases, namely
\begin{eqnarray}
\label{eq:n=123}
N=\bma 1&1\\1&-1 \ema,
\quad
\bma 1&-1\\1&-1 \ema,
\text{\quad or \quad}	
\bma 1&-1\\-1&1 \ema.
\end{eqnarray}
For the first matrix in \eqref{eq:n=123}, the orthogonality of row $2,3$ of $M$ in \eqref{eq:222223=case2-iii} implies that $m_{26}=-m_{36}=1$. Lemma \ref{le:properties} (ii.b) implies that column $1,6$ of $M$ are not orthogonal, and it is a contradiction with the fact that $M$ is a CHM. For the second matrix in \eqref{eq:n=123}, one can show that row $2,3$ of $M$ are not orthogonal, and it is again a contradiction with the above-mentioned fact. For the third matrix in \eqref{eq:n=123}, the orthogonality of rows and columns of $M$ show that $w_2,w_3,w_4,w_5$ are $i$ or $-i$, and the lower-right $3\times3$ submatrix of $M$ in in \eqref{eq:222223=case2-iii} has imaginary elements in either $m_{64},m_{65},m_{46},m_{56}$ or $m_{65},m_{66},m_{56}$. These facts imply that $w_1=iw_2$ or $-iw_2$ is real. It is a contradiction with the fact that $w_1$ is imaginary.  
\end{proof}

\section{Conclusions}
\label{sec:con}
We have analytically obtained the number of real entries of an $n\times n$ complex Hadamard matrix (CHM) with $n=2,3,4,6$. We also have partially characterized the number for general $n$. The first main result is that the number can be any one of $0-22,24,25,26,30$ for $n=6$. Applying our result to the existence of four MUBs in dimension six. we have shown that the number of real entries in any CHM of an MUB trio does not exceed the upper bound $22$. An open problem is to reduce the upper bound.

\section*{Acknowledgements}

Authors were supported by the  NNSF of China (Grant No. 11871089), and the Fundamental Research Funds for the Central Universities (Grant Nos. KG12080401 and ZG216S1902).

\bibliographystyle{unsrt}

\bibliography{mengfan=6x6hadamardreal}
\end{document}